\RequirePackage{fix-cm}
\RequirePackage{amsmath}
\documentclass[smallextended]{svjour3}       
\smartqed  

\usepackage{graphicx,tikz}
\usepackage{booktabs}
\usepackage{xcolor}
\usepackage{colortbl}
\usepackage{url}
\journalname{
Autonomous Agents and Multi-Agent Systems}
\begin{document}

\title{Preventing Social Disappointment in Elections
}


\author{Mohammad Ali Javidian\and Pooyan Jamshidi\and Marco Valtorta         \and
        Rasoul Ramezanian 
}


\institute{Mohammad Ali Javidian, Pooyan Jamshidi, and Marco Valtorta \at
              University of South Carolina, Columbia, SC, USA. \\
              Tel.: +1 (803) 777-4641\\
              Fax: +1 (803) 777-3767\\
              \email{javidian@email.sc.edu, \{pjamshid,mgv\}@cse.sc.edu}           
          \and
          Rasoul Ramezanian \at
              Ferdowsi University of Mashhad, Mashhad, Iran\\
              \email{rasool.ramezanian@gmail.com}
}

\date{Received: date / Accepted: date}

\maketitle

\begin{abstract}
Mechanism design is concerned with settings where a policy maker (or social
planner) faces the problem of aggregating the announced preferences of multiple agents into
a collective (or social), system-wide decision. One of the most important ways for aggregating
preference that has been used in multi-agent systems is election. In an election, the aim is to select the
candidate who reflects the common will of society. Despite the importance of this
subject, in some  situations, the result of the
election does not respect the purpose of those who execute it and the election leads
to dissatisfaction of a large amount of people and in some cases causes polarization in societies. To analyze these situations, we introduce a new notion called social disappointment and we show which voting rules can prevent it in elections. In addition, we propose new protocols to prevent social disappointment in elections. A version of the impossibility theorem is proved
regarding social disappointment in elections, showing that there is no voting rule for four or more candidates that simultaneously satisfies avoiding social disappointment and Condorcet winner criteria.
We give conditions under which one of our new protocols always selects the Condorcet winner under the assumption of single peakness. 
We empirically compare our protocols with seven well-known other voting protocols and  observe that our protocols are capable of preventing social disappointment and are more robust against manipulations.
\end{abstract}

\section{Introduction}
\label{intro}
Social choice theory is concerned with the design and analysis of methods for
collective decision making \cite{suzumura}. Voting procedures are among the most important
methods for collective decision making. Voting procedures focus on the aggregation of
individuals' preferences to produce collective decisions. In practice, a voting procedure is
characterized by ballot responses and the way ballots are tallied to determine winners. Voters
are assumed to have clear preferences over candidates (alternatives) and attempt to maximize
satisfaction with the election outcome by their ballot responses. 

In social choice theory, it is emphasized that the result of an election must reflect the general will of voters. For this reason, in social choice theory, there are different and various criteria to address the efficiency and desirability of a voting procedure. Voting procedures are formalized by social choice functions, which map ballot response profiles into election
outcomes \cite{bramsfishburn}. However, despite the fact that this is an important matter, in the real world situations sometimes the winner of an election is the one who causes a vast dissatisfaction in the society and in some cases leads to polarization \cite{GF}. For example, according to Reilly \cite{Reilly}, in societies with deep ethnic cleavages such as Papua New Guinea, Sri Lanka, Northern Ireland, Estonia and Fiji, elections can encourage extremist ethnic appeals, zerosum political behaviour and ethnic conflict, and consequently often lead
to the breakdown of democracy. Reilly in \cite{Reilly} and Grofman and Feld in \cite{GF} examine the potential of \textit{electoral engineering} as a mechanism of conflict management in divided societies, and show that avoiding the selection of extremist candidates with substantial first round but little overall support is a central merit of voting rules like alternative vote and Coombs method  as compared to plurality.  As a result, it is desirable that the alternative defeated by a majority in pairwise contests against any
other alternatives  (i.e., Condorcet loser), and also the alternative that is at the bottom of at least half of the
individual preference profiles should not be elected-both to make less likely the election of candidates with limited overall support and to mitigate ethnic conflict in divided societies.

To avoid social disappointment in voting mechanisms, this paper introduces new criterion called \textbf{social disappointment}. In fact, we consider three factors relevant to picking (and designing) a voting rule to be used to select a candidate(s) from among a set of choices: (1) avoidance of Condorcet loser, (2) avoidance of social disappointment, (3) resistance to manipulability via strategic voting. Based on these factors, we analyze some well-known voting protocols and we show which ones avoid of Condorcet loser and which ones prevent social disappointment in elections. For example, since procedures like Borda and Copeland do not elect Condorcet loser in elections \cite{RobinsonUllman}, we recommend them to avoid of Condorcet losers. Also, we show that  Coombs method prevent social disappointment in elections. Furthermore, we design two new voting rules called  Least Unpopular (LU) and Least Unpopular Reselection (LUR) Procedures that both of them avoid  social disappointment in voting, while retaining monotonicity criterion as opposed to Coombs. Moreover, in elections, it is reasonable for a voter to rank order the candidates in order of proximity to the voter's ideal point if voters are required to submit a
rank-ordered ballot \cite{GF}. This assumption, in the social choice literature, is called \textit{single-peaked preferences}. Rather than considering the properties of voting rules in the abstract, we evaluate them in the politically realistic situations where voters are posited to have \textit{single-peaked preferences} over alternatives. We prove that, under the assumptions of single-peaked preferences and no party holding a majority of first place preferences, for four alternatives or fewer, the LUR rule will always select the Condorcet winner. A version of the impossibility theorem is stated and proved regarding to the notion of social disappointment, showing that there is no voting rule for four or more candidates that simultaneously satisfies avoiding social disappointment and Condorcet winner criteria. In order to address the robustness of discussed voting rules in this paper against manipulation, we design four scenarios of manipulation such as control and bribery, and evaluate the resistance of our proposed protocols i.e., LU and LUR against manipulations.

\textbf{Summary of results.}
Based on experimental results, as we expected, social disappointment does not happen for Coombs, LU ,and LUR methods. Despite the fact that Copeland and Borda cannot avoid social disappointment, both of them lead to social disappointment in fairly small number of cases, which is also consistent with theoretical results.  The best performance against manipulation via strategic voting belongs to the LU and LUR in all scenarios, showing that these procedures are more robust against manipulation. In all cases, plurality has the worst performance. 

\textbf{Novelty.}
The central point of this paper is the claim that some of the same arguments that have proposed to justify replacing plurality with alternative vote (a.k.a. instant runoff) \cite{BowlerGrofman,Horowitz,Reilly} apply with greater or equal force when replacing plurality with the Coombs, LU, or LUR procedures. Grofman and Feld in \cite{GF} made the argument that electoral reformers who have been advocating the alternative vote (AV)/instant runoff need to take a serious look at its close relative, the Coombs rule. In this paper, we argue that the LU and LUR procedures are directly comparable with Coombs with respect
to one of our three criteria i.e., avoidance of social disappointment and both of them are superior to Coombs with respect to one of the three criteria i.e., resistance to  manipulability via strategic voting.

Our main contributions are the following:
\begin{itemize}
    \item We introduce new criteria called social frustration (SF) and social disappointment (SD) respectively, and we show that which protocols avoid SF and SD in elections.
    \item We provide new voting protocols named the least unpopular (LU) and the least unpopular reselection
    procedures (LUR), and we show that both of them prevent SD in elections.
    \item We prove that under the assumptions of single-peaked preferences and no party holding a majority of first place preferences, when we have four parties or fewer, the LUR procedure always picks the Condorcet winner.
    \item We provide and prove an impossibility theorem by considering the notion of social disappointment.
    \item We experimentally compare the numerical results obtained from the implementation of nine voting rules, and we show that Coombs, LU, and LUR are capable to prevent SD, and also more successful than the other procedures regarding their resistance against manipulation in different scenarios. An R language package  that implements nine voting procedures and four manipulation scenarios, all described in this paper, along with the resulting experimental data are available at: \textcolor{blue}{\url{https://github.com/majavid/AAMAS2019}}.
\end{itemize}

\section{Basic Definitions and Concepts}
In this paper we consider elections with more than two candidates.
In order to define some central concepts such as ballot, profile, voting rule, and so on, we need to consider the following basic definitions and notations.

\textbf{Preference Relations.} In set theory $|A|$ denotes the number of elements in the finite set A. Any subset R of A $\times$ A is a binary relation on A, and in
this case we write "aRb" to indicate that (a, b) $\in$ R, and
we write "$\neg$(aRb)" to indicate that (a, b) $\not\in$ R. The binary relations we are most concerned with satisfy one or more of the following properties.
\begin{definition}\label{def1}
A binary relation R on a set A is:
\begin{equation*}
\begin{array}{lcl}
reflexive&if&\forall x\in A, xRx\\
symmetric&if&\forall x,y\in A, \textrm{ if } xRy \textrm{ then } yRx\\
asymmetric&if&\forall x,y\in A, \textrm{ if } xRy \textrm{ then } \neg(yRx)\\
antisymmetric&if&\forall x,y\in A, \textrm{ if } xRy \textrm{ and } yRx \textrm{ then } x=y\\
transitive&if&\forall x,y,z\in A, \textrm{ if } xRy \textrm{ and } yRz \textrm{ then } xRz\\
complete&if&\forall x,y\in A, \textrm{ either } xRy \textrm{ or } yRx \textrm{ (or both)}.
\end{array}
\end{equation*}
Also, a binary relation R on a set A is a \textit{weak ordering} (of A) if it is
transitive and complete and a \textit{linear ordering} (of A) if it is also antisymmetric. If R is a weak ordering of A, then the derived relations of
\textit{strict preference} P and \textit{indifference} I are arrived at by asserting that xPy iff $\neg$(yRx) and xIy iff xRy and yRx.
\end{definition}

\begin{definition}\label{def3}
If A is a finite non-empty set (which we think of as the set
of alternatives (candidates) from which the voters are choosing), then an A-ballot is a weak
ordering of A. If, additionally, n is a positive integer (where we think of N =$\{1, \dots , n\}$
 as being the set of voters), then an (A, n)-profile is an n-tuple of
A-ballots. Similarly, a linear A-ballot is a linear ordering of A, and a linear
(A, n)-profile is an n-tuple of linear A-ballots.
\end{definition}

The following definition collects some additional ballot-theoretic notation
we will need.
\begin{definition}\label{def4}
Suppose P is a linear (A, n)-profile, X is a set of alternatives
(that is, X $\subseteq$ A), and i is a voter (that is, i $\in$ N). Then:
\begin{equation*}
    \begin{array}{lcl}
      top_i(P)=x& iff & \forall x\in A, xR_iy \\
      max_i(X,P)=x& iff & x\in X, \forall y\in X: xR_iy\\
      min_i(X,P)=x& iff & x\in X, \forall y\in X: yR_ix
    \end{array}
\end{equation*}
\end{definition}

\textbf{Elections and desirable properties.} In the context of social choice theory, we assume that there is a set A whose elements are called alternatives (or
candidates)\footnote{We use "alternatives" and "candidates" interchangeably.} and typically denoted by a, b, c, etc. There is also
a set P whose elements is called people (or voters). Each person p in P
has arranged the alternatives in a list (with no ties) according to preference. Such a list will be called an individual preference list, or, for brevity, a ballot. A sequence of ballots is called a profile. The following definition uses the concepts of weak and linear orderings to formalize these election-theoretic terminologies. In this paper, we only consider elections with more than two candidates.
A voting procedure is a special kind of function where a typical input is a profile and an output is a single alternative, or a
single set of alternatives if we allow ties, or "NW" indicating that there
is no winner. Because of the importance of this notion, we record it here formally
as a definition, along with six desirable properties for comparison of different
voting rules.
\begin{definition}\label{def5}
Suppose that A is a non-empty set, n is a positive integer, N=$\{1,\dots,n\}$, and V is a function whose domain is the collection of all (A, n)-profiles.
Then V is a \textit{voting rule}\footnote{We use "voting rule", "voting procedure", and "voting protocol" interchangeably.} for (A, n) if, for every (A, n)-profile P, the election outcome
V(P) is a  subset of A. We say that V satisfies:
\begin{itemize}
    \item Always-a-winner Condition (AAW), if V(P)$\ne \emptyset$. In words, a voting procedure satisfies AAW condition if for every sequence of individual preference
    lists, the procedure outputs at least one winner.
    \item Condorcet winner condition (CWC) if 
    for every $x\in A$: $$x\in V(P) \textrm{ iff }\forall y\in A, W(x,y,P)\ge W(y,x,P),$$
    where $W(x, y, P)$ denotes the
    number of voters who rank x over y on their ballot in P i.e., $W(x, y, P)=|\{i\in N|xP_iy\}|.$
    Here, $P_i$ is a \textit{strict preference} on A such that $xP_iy$ iff $\neg(yR_ix)$, where $R_i$ is a binary relation on $A$. An alternative x is said to be a Condorcet winner
if it is the unique winner in Condorcet's method.
    \item Pareto condition, if for every pair of alternatives x, y $\in$ A, if x$P_i$y for every i, then y $\not\in$ V(P).
    A voting procedure is said to satisfy Pareto condition (or just
Pareto) if for every pair of x and y, all the voters
prefer x to y, then y is not a social choice. 
    \item Monotonicity (or V is monotone) if x $\in$ V(P), i $\in$ N and y $\in$ A with y P$_i$ x, assume that the profile Q exists such that P$|_{N- \{i\}}$ =Q$|_{N-\{i\}}$\footnote{If R
    is a binary relation on A and S $\subseteq$ A, then the \textit{restriction} of R to S, denoted R$|_S$,
    is the binary relation on S given by R$|_S$ = R $\cap$ (S $\times S$).}, P$_i|_{A- \{x,y\}}$ =Q$_i|_{A-\{x,y\}}$, and x Q$_i$ y then x $\in$ V(Q). A voting procedure is said
to be monotone if x is
a social choice and someone changes his or her preference list
by moving x up, then x should still be the social choice. 
    \item Independence of irrelevant alternatives (IIA) condition, if for every pair of
    (A, n)-profiles P and P', and every pair of alternatives x, y $\in$ A, if x $\in$ V(P) and y $\not\in$ V(P) and $R_i|_{\{x,y\}}=R'_i|_{\{x,y\}}$ for every i, then y $\not\in$ V(P'). A voting procedure satisfies independence of irrelevant alternatives (IIA) condition whenever the social preferences between alternatives x and y depend only on the individual preferences between x and y. The condition of independence
of irrelevant alternatives was first used by Arrow  in 1951 \cite{A1}. We say that an alternative is a \textit{Condorcet loser} if it would be defeated
by every other alternative in a kind of one-on-one contest that
takes place in a sequential pairwise voting with a fixed agenda\footnote{When we speak of a "fixed agenda,"
we are assuming we have a specified ordering of the alternatives.}.
    \item Condorcet loser criterion (CLC), if there exists x $\in$ A such that for all y $\in$ A, W(x,y,P) $<$ W(y,x,P), where $W(x, y, P)=|\{i\in N|xP_iy\}|$, then x $\not\in$ V(P)\footnote{From the definition of the Condorcet loser it is clear that if existed would be unique.}. Further,
we say that a social choice procedure satisfies the Condorcet loser
criterion (CLC) provided that a Condorcet loser is never among the social
choices. (see \cite{Taylor,Taylor2} for more details and examples).
\end{itemize}
\end{definition}

\textbf{Voting procedures.} In this paper, we consider nine examples of social choice procedures such as Condorcet's method \footnote{With Condorcet's method, an alternative x (if any) is among the winners
if for every other alternative y, at least half of the voters rank x over y
on their ballots. Our usage of the term \textit{Condorcet's method} follows \cite{Taylor}.}, Plurality rule, Hare system \footnote{This procedure was introduced in 1861 by Thomas Hare \cite{Taylor2}.} (is also known by names such as the "single transferable vote system" (stv) or "instant runoff voting"), Borda count \cite{Borda}, sequential pairwise voting with a fixed 
agenda (Seq. Pairs), Copeland \footnote{This procedure was introduced in
an unpublished 1951 note by A. H. Copeland \cite{Taylor2}.}, Coombs \cite{Coombs} and dictatorship.  For formal definitions and examples, see \cite{Taylor,Taylor2}. 
Here we just define Hare and Coombs procedures formally. In order to define these procedures we need the following notation and definition:

Suppose that $A$ is a set of alternatives, $n$ is a positive integer, and $V$ is a voting rule defined not for just $(A, n)$, but for $(A', n)$ for every $A'\subseteq$ A. Now, for every $(A, n)$ profile P, we can consider the
sequence $\langle W_1,\dots,W_{|A|}\rangle$, where $W_1 = V(P), W_2 = V(P|_{W_1}), W_3 = V(P|_{W_2})$, etc. Notice that 
\begin{itemize}
    \item [(i)] $A\supseteq W_1\supseteq W_2\supseteq \dots\supseteq W_{|A|}$, and
    \item[(ii)] if $W_j = W_{j+1}$, then $W_{j+1} =\dots = W_{|A|}$.
\end{itemize} 

The Hare procedure and the Coombs procedure are special cases of the
general idea of repeatedly using a single procedure to break ties among winners. 
\begin{definition}\label{def:Hare}
One repeatedly deletes the alternative or the alternatives
with the fewest first-place votes, with the last group of alternatives to be deleted
tied for the win. More precisely, $V$ is the Hare voting rule (also called "the
Hare system" or "the Hare procedure") if $V = V^*_H$ where $V_H(P)$ is the set of all
alternatives except those with the fewest first-place votes in $P$ (and all tie if all
have the same number of first-place votes).
\end{definition}

\begin{definition}\label{def:Coombs}
One repeatedly deletes the alternative or the alternatives with the most last-place votes, with the last group of alternatives to be
deleted tied for the win. More precisely, V is the Coombs voting rule (also called
"the Coombs procedure") if V = V$^*_C$ where V$_C$(P) is the set of all alternatives
except those with the most last-place votes in P (and all tie if all have the same
number of first-place votes).
\end{definition}

\textbf{Single-peaked preferences assumption}. It is well-known that certain domain restrictions enable the circumvention of impossibility theorems and can make computationally difficult problems easy~\cite{Brandtetal2016}. Arguably the most
well-known of these domain restrictions is Black's single-peakedness \cite{Black}. Preferences are said to be single-peaked
with respect to the order $<$ on the alternatives if the following holds: for every voter, as we move
away (according to $<$) from the voter's most-preferred alternative, the alternatives
will become less preferred for that voter.
If we sort voters according to their most preferred alternative (breaking ties arbitrarily), then the ((n+1)/2)th voter is called
the \textit{median voter}, where $n$ is the odd number of voters. His top choice is always identical to the Condorcet winner,
as was first observed by Black \cite{Black}.
An application of ideological preferences is to think of the outcome space as different policies in an ideological spectrum: policies from the Left vs policies from the Right; policies that are more liberal vs policies that are more conservative; policies that are pro free markets vs policies that are pro state intervention. Voters have single-peaked preferences if they have an ideal balance between the two directions of the ideological spectrum and if they dislike policies the farther away they are from their ideal point.
\begin{definition}
Let a $(A,n)$-profile $P$  and a set of voters $N$ be given.  We define the preference profile for $P$, denoted $PP$, as follows:
$PP$ is a function from each alternative in $A$ to the non-negative integers.  
A sequence of individual preference lists is said to satisfy the single peaked preferences assumption if there exists a function $PP$ such that, for each voter $i \in N$,
the graph of the application of the function $PP$ to each alternative has a single peak.
It is well known that a preference profile $PP$ is single-peaked if, for every $x, y, z\in A$, it holds that, if $(PP(x) < PP(y) < PP(z))$ or $(PP(z)< PP(y) < PP(x))$, then $x \succ_i y$ implies $y \succ_i z$ for every $i \in N$, where $\succ_i$ is the preference ordering for voter $i$.
\end{definition}

\section{Condorcet Loser and Social Disappointment in Voting Systems}\label{SFSD}
To avoid vast dissatisfaction and polarization in society because of the outcome of elections, we consider two intuitively compelling requirements that one may impose on voting procedures: (1) the alternative(s) defeated by a majority in pairwise contests against any other alternatives (Condorcet loser) should not be elected, and (2)  the alternative that is ranked at the bottom of at least half of the individual preference profiles should not be elected. In order to deal with these situations, we introduce new notion called \emph{social disappointment} in voting systems.

\begin{definition} 
(\textit{Social disappointment}): Suppose that A is a non-empty set (the set of alternatives), n is a positive integer (the number of voters), N=$\{1,\dots,n\}$ (the set of voters), and V is a voting rule for (A, n). If there exist x $\in$ A such that $|\{i\in N| \min_i(A,P)=x\}|\ge \frac{n}{2}$ and x $\in$ V(P) we say that \textit{social disappointment} (SD) in voting has occurred. Also, we say that V satisfies the \textit{social disappointment criterion} (SDC), if there exist x $\in$ A such that $|\{i\in N| \min_i(A,P)=x\}|\ge \frac{n}{2}$, then x $\not\in$ V(P). 
\label{def:sd}
\end{definition}
In words, social disappointment in voting happens when the outcome of an election (for 3 or more alternatives) includes an alternative which is at the bottom of at least half of the individual preference lists.
If in the definition of social disappointment instead of using the expression "at least half of" one uses the phrase "more than the half of", then a new definition of social disappointment will be obtained called \textit{strict social disappointment}. We leave this to the reader to verify that the Condorcet's method, the Borda count rule, the Seq. Pairs., and the Coombs' procedure can avoid strict social disappointment. Also, the occurrence of strict social disappointment will result in social disappointment. However, the reverse may not hold true.

In order to illustrate the key concepts Condorcet loser and SD, consider the following example\footnote{We have adapted this example from \cite{BCE}.}.
\begin{example}\label{ex1} 
Consider the following situation in which there are four Dutchmen, three Germans, and two Frenchmen who have to decide on which drink to be served for lunch (only a single drink will be served to all). 
  \begin{table}[ht]
  \label{t:ex1}
  \centering
  \begin{tabular}{ccc}
    \toprule
    Voters 1-4&Voters 5-7&Voters 8 and 9\\
    \midrule
    Milk & Beer&Wine\\
    Wine & Wine& Beer\\
    Beer & Milk & Milk\\
  \bottomrule
\end{tabular}
\end{table}

Now, which drink should be served based on these individuals' preferences? Milk
could be chosen since it has the most agents ranking it first. Milk is the winner
according to the plurality rule, which only considers how often each alternative is
ranked in the first place. However, the majority of agents will be dissatisfied with this
choice as they prefer any other drink to Milk. For such an occasion in terms of social choice theory one can say that the Condorcet Loser-in this example milk-is the social choice and this, on its own, is one of the undesirable situations in the social choice theory. 

Now to look at it from another perspective, Milk is the alternative which is at the bottom of more than half of the voters' preferences lists i.e., social choice is an alternative with the least social support and has the most social dissatisfaction or to be even more serious has the most social resentment. In other words, social disappointment has occurred. 
\end{example}

\subsection{Condorcet Loser and Voting Systems}\label{I6}
In this subsection, we provide an example that shows the likelihood of plurality rule to electing Condorcet loser is not very unusual and, indeed, happens in real-life political elections. Then, we argue that which voting rules are able to prevent choosing a Condorcet loser.

Grofman and Feld \cite{GF} pointed out that picking the Condorcet loser  as the winner is a poor voting method, indeed. Plurality rule can have this flaw \cite{BCE}. To show its importance, we provide an important historical voting situation where resulted in social disaster by electing a Condorcet loser as a winner of an election.

In a well-known US senatorial contest in the State of New York in 1970, the candidate from the Conservative party (Buckley) was a Condorcet Loser because he would have lost in one-to-one general election contest with
either the Democrat (Otinger) or the liberal Republican (Goodell). Yet Buckley won the general election with a plurality
vote despite his more liberal opponents receiving 60\% of the
vote. Although the liberal Republicans were a minority among the Republican
voters, most liberal Republicans preferred Goodell to Buckley, and in a general
election putting Buckley against Goodell, they would have been joined by a
high proportion of the Democrats who would also have clearly preferred
Goodell to Buckley. On the other hand, while Otinger might not have done as well as Goodell among liberal Republicans
in a one-to-one contest with Buckley, he would have made up for that by
getting virtually all the Democratic vote. We can see this as a situation involving
single-peaked preferences where the two liberal candidates (Goodell and
Otinger) split the liberal vote, allowing the least preferred choice among a
majority of the voters to win (see \cite[page 647]{GF}).

Therefore, in order to prevent  choosing a Condorcet loser in elections, we recommend to use the voting protocols such as Condorcet, Borda, Copeland, or Seq. Pairs, because all of them satisfy the Condorcet loser condition \cite[Chapter 4]{RobinsonUllman}\footnote{The authors of \cite{RobinsonUllman} use the term of anti-Condorcet
criterion instead of CLC.}.

\subsection{Social Disappointment in Voting Systems}\label{I7}
In this section, we justify the importance of the notion of social disappointment in real-world political elections, and show which voting protocols avoid SD. Then, we show that Condorcet loser and social disappointment are two distinct concepts and the occurrence of none of them implies the occurrence of the other one.

Some civilizations live in a way that the racial/ethnic/religious divisions can potentially cause polarization in the society. In situations where centrism is defined in terms of conciliatory views, the necessity of choosing the moderate candidates and keeping away from the extremist candidates-who are not supported by a broad spectrum of people-are shown \cite{GF}. Similarly, we design a mechanism that satisfies these properties and avoid SD in elections.

The main question here is that whether any of the given protocols in this article prevent SD in voting systems? The voting procedure that prevents social disappointment in voting systems is the \textit{Coombs method} (Coombs rule), which was introduced by the famous psychologist Clyde Coombs \cite{Coombs}.

Suppose that there is a candidate that is at the end of at least half of the preference
profiles. The Coombs' procedure deletes this alternative from the profile lists in the
first stage and does not let this candidate be elected, so social disappointment can not occur for this alternative. Coombs procedure satisfies the properties in Table~\ref{tab:properties}.
\begin{table}[ht]
\caption{Distinctive characteristics of Coombs method.}
\label{tab:properties}
\centering
\begin{tabular}{|c|c|c|c|c|c|c|c|}
\hline
 & {AAW} & {CWC} & {Pareto} & {Mono}& {IIA} & {SDC} &{CLC}\\
\hline
Coombs & YES & No & YES & No & No & YES & No\\
\hline
\end{tabular}
\end{table}

Grofman and Feld  \cite{GF} have shown that many conclusions suggested as a justification for replacing the plurality rule with the Hare system can be used as even stronger and more accurate evidences for Coombs' procedure. They agree that the supporters of the Hare system, like Donald Horowitz, have proven that when the Hare system is used the probability of a moderate candidate's winning the election instead of an extremist is higher than that when the plurality rule is used. A \emph{moderate candidate} here refers to the one who has the support of most of the voters but is not their first priority, and an \emph{extremist} here refers to the candidate who is the first priority of many voters but loses contests in one or more one-on-one competitions with the other candidates. The influence and permeation of Horowitz's idea led to its adoption in places such as Fiji and Papua New Guinea \cite{Reilly,GF}. 
To prevent polarization In societies with racial, religious and ethnic conflicts, election organizers are advised to help reduce social dissatisfaction and encourage the election of a moderate politician by using protocols like the Coombs procedure.

The following proposition shows that Condorcet loser and social disappointment in voting systems are two distinct concepts and none of them implies the other one.

\begin{proposition} 
The definitions given for Condorcet loser and social disappointment in voting systems are completely different; the occurrence of one of them never guarantees that of the other.\label{p1}
\end{proposition}
\begin{proof} 
Consider the following profile:
\begin{table}[ht]
\centering
\begin{tabular}{cccc}
\toprule
Voter 1 & Voter 2 & Voter 3 & Voter 4 \\
\midrule
a & a & b & c \\
c & b & c & b \\
b & c & a & a \\
\bottomrule
\end{tabular}
\end{table}

Applying the plurality rule, $a$ is the socially-selected candidate and SD will occur. However, $a$ is not the Condorcet loser, since it does not lose in a pairwise contest with $b$; in other words, $W(a,b,P)=|\{1,2\}|=2\not< W(b,a,P)=|\{3,4\}|=2.$

Now, consider the following profile:
\begin{table}[ht]
\centering
\begin{tabular}{ccc}
\toprule
32 voters & 38 voters & 10 voters \\
\midrule
b & c & b\\
a & a & c \\
c & b & a \\
\bottomrule
\end{tabular}
\end{table}

Obviously, $a$ is the Condorcet loser but, he is not at the end of at least half of the individual preference lists. Therefore, if one protocol (for instance LU protocol which will be introduced in the next subsection) considers $a$ as the winner of the social selection, despite its being the Condorcet loser, the SD will not occur.
\end{proof}

\subsection{The Least Unpopular (LU) and the Least Unpopular Reselection (LUR)}
First, we show that plurality, Borda, Condorcet, Copeland, Seq. Pairs., and Hare do not satisfy SDC.
\begin{proposition}
Plurality, Borda, Condorcet, Copeland, Seq. Pairs., and Hare method do not satisfy SDC.
\end{proposition}
\begin{proof}
For each voting procedure we have:
\begin{itemize}
    \item \textbf{Plurality}: see Example \ref{ex1} that shows plurality rule violates SDC.
    \item \textbf{Borda}: Borda rule fails to satisfy SDC. More precisely, the Borda count rule can always prevent SD in voting except in one case. In this case the social choice set will consist of all the alternatives. For example, consider the three alternatives 'a', 'b', and 'c' and the following sequences of two preference lists:
\begin{center}
\begin{table}[!hbp]
\centering
\begin{tabular}{cc}
\toprule
Voters 1 and 2& Voters 3 and 4\\
\midrule
a& c \\
b& b \\
c & a \\
\bottomrule
\end{tabular}
\end{table}
\end{center}
The alternatives 'a', 'b' and 'c' are the social choice when the Borda count procedure is used, but 'a' (also 'c'), is at the bottom of half of individual preference lists, and so social disappointment has taken place.

Note that there are 'k' candidates $(k\geq 3)$ and 'n' voters $(n\geq 3)$. The total sum of scores in Borda count rule is equal to:
$$n((k-1)+(k-2)+...+2+1+0)=\frac{n(k-1)k}{2}$$
Now consider that 'n' is an odd number, without loss of generality, and also consider that $x_1$ is a social choice and there is SD in voting, so $x_1$ must be at least the last preference in $\frac{n+1}{2}$ of individual preferences lists. Now consider the most optimistic possibility that in $\frac{n-1}{2}$ of the remaning lists $x_1$ is at the top. Thus Borda score for the alternative $x_1$ equals: $\frac{(n-1)(k-1)}{2}$. Now if this amount is subtracted from the whole Borda score it gives:
$$\frac{k(k-1)n}{2}-\frac{(k-1)(n-1)}{2}=\frac{(k-1)(kn-k+1)}{2}$$
Now if, in the most optimistic possibility, the remaining score is again shared among the other $k-1$ candidates equally, the amount of Borda score for every other candidate is $\frac{(k-1)n+1}{2}$ which clearly is more than Borda score for $x_1$, and this is against $x_1$ being the social choice. In the end if the number of voters is an odd number and voting is done according to Borda count rule, SD will definitely not occur.

Now consider that 'n' is an even number, without loss of generality and also consider that $x_1$ is a social choice and there is SD in voting. So, in the most optimistic possibility, Borda score for $x_1$ is equal to: $\frac{n(k-1)}{2}$. Now, if this amount is subtracted from the whole Borda score gives:
$$\frac{k(k-1)n}{2}-\frac{(k-1)n}{2}=\frac{(k-1)n(k-1)}{2}$$
If, in the most optimistic possibility, the remaining score again shared among the other $k-1$ candidates equally, the amount of Borda score for every other candidate would be $\frac{(k-1)n}{2}$ which is clearly equal to the Borda score for $x_1$ and thus the social choice set consists of all the candidates. Otherwise, if the remaining score is shared among every other candidate, $x_1$ can no longer be the social choice according to Borda count rule. The reason is that there is at least one candidate that has a score higher than that of $x_1$.
    \item \textbf{Copeland, Condorcet, and Seq. Pairs.}: Consider the following profile (each column shows a ballot of each voter):
    \begin{center}
\begin{table}[ht]
\centering
\begin{tabular}{cccccc}
\toprule
d & d & d & c & b & b\\
a & a & c & a & c & c \\
b & b & a & b & a & a \\
c & c & b & d & d & d \\
\bottomrule
\end{tabular}
\end{table}
\end{center}
The alternative 'd' is the unique social choice when the Condorcet method is used, but 'd' is at the bottom of half of the individual preference lists, and so social disappointment has taken place. All alternatives are in the social choice set when the Copeland method is used.  Again, social disappointment has taken place. Also, suppose that alphabetic ordering of the alternatives is the agenda when the Seq. Pairs method is used. So, $\{c,d\}$ is the social choice, and, again, SDC is violated.
\item \textbf{Hare}: Consider the three alternatives 'a', 'b', and 'c' and the following sequence of ten preference lists grouped into voting blocks of size four, three, and one:
\begin{center}
\begin{table}[ht]
\centering
\begin{tabular}{ccc}
\toprule
Voters 1-4& Voters 5-7& Voters 8-10\\
\midrule
a& c & b \\
b& b & c \\
c & a & a \\
\bottomrule
\end{tabular}
\end{table}
\end{center}
The alternative 'a' is in the social choice set when the Hare system is used. Although 'a' is the social choice, it is at the bottom of more than half of the individual preference lists and so social disappointment has taken place.
\end{itemize}
\end{proof}

Of the seven voting procedures we
considered, only Coombs's method satisfies SDC. However, Coombs does not satisfy the  monotonicity criterion, which formalizes the crucial idea that increased support for a candidate never hurts, and may help her to win. Otherwise, voters would be afraid to cast their ballots in an \emph{honest way}, aware that
a vote for their sincere first choice could harm the cause of electing her. We, therefore, want to pursue a voting procedure to satisfy  the monotonicity criterion \cite{RobinsonUllman}.  Suppose we are seeking a voting procedure that satisfies SDC and monotonicity simultaneously.

Now, we introduce a new voting protocol that satisfies monotonicity and prevents SD in voting systems. 
\begin{definition}(\textit{The Least Unpopular (LU) procedure})
The social choice(s) in the least unpopular procedure (LU) is (are) the alternative(s) that appear(s) the fewest number of times at the bottom of individual preference lists. More precisely, V is the LU procedure if V(P)=$\{x\in A|lp(x) \textrm{ is minimum} \}$, where $lp(x)=\{i\in N| \min_i(A,P)=x\}$.\label{lu}
\end{definition}

This protocol satisfies the AAW, Monotonicity, and social disappointment criterion, but does not satisfy the CWC, CLC, Pareto, and IIA criterion.
\begin{proposition}
The Least Unpopular procedure does not satisfy the CLC, CWC, IIA, and Pareto criteria.\label{lu1}
\end{proposition}
\begin{proof} 
Consider the four alternatives 'a', 'b', 'c', and 'd' and the following profile:
\begin{center}
\begin{table}[ht]
\centering
\begin{tabular}{ccc}
\toprule
Voters 1 and 2& Voter 3& Voter 4\\
\midrule
a & c & d \\
b & a & a \\
c & b & b \\
d & d & c\\
\bottomrule
\end{tabular}
\end{table}
\end{center}
The alternatives 'a' and 'b' are the social choices when the Least Unpopular procedure is used. Thus, the alternative 'b' is in the set of social choices even though everyone prefers 'a' to 'b'. This shows that the Pareto criterion fails.
Now consider the three alternatives 'a', 'b', 'c' and the following profile:
\begin{center}
\begin{table}[ht]
\centering
\begin{tabular}{cc}
\toprule
Voters 1 and 2& Voter 3\\
\midrule
a & b \\
b & c \\
c & a \\
\bottomrule
\end{tabular}
\end{table}
\end{center}
The alternative 'b' is the social choice when the Least Unpopular procedure is used. However, 'a' is clearly the Condorcet's winner, defeating each of the other alternatives in one-on-one competitions. Since the Condorcet's winner is not the social choice in this situation, it is clear that the Least Unpopular procedure does not satisfy the Condorcet's winner criterion. In other words, the alternative 'b' is a non-winner. Now suppose that voter 3 changes his or her list by interchanging the alternatives 'a' and 'c'.
The lists then become:
\begin{center}
\begin{table}[ht]
\centering
\begin{tabular}{cc}
\toprule
Voters 1 and 2& Voter 3\\
\midrule
a & b \\
b & a \\
c & c \\
\bottomrule
\end{tabular}
\end{table}
\end{center}
Notice that the alternative 'b' is still above 'a' in the third voter's list. However, the Least Unpopular procedure now has 'a' and 'b' tied as the winner. Thus, although no one changed his or her preference regarding the alternatives 'a' and 'b', the alternative 'a' changed position from being a non-winner to being a winner. This shows that the independence of irrelevant alternatives fails in the Least Unpopular procedure.

Now, consider the following profile:
\begin{table}[ht]
\centering
\begin{tabular}{ccc}
\toprule
32 voters & 38 voters & 10 voters \\
\midrule
b & c & b\\
a & a & c \\
c & b & a \\
\bottomrule
\end{tabular}
\end{table}
Obviously, $a$ is the Condorcet loser but LU protocol considers $a$ as the winner of the social selection. This shows that LU does not satisfies CLC.
\end{proof}

\begin{proposition}
The Least Unpopular procedure satisfies the SDC, AAW, and Monotonicity criterion.\label{lu2}
\end{proposition}
\begin{proof}
Since $|A|$, the number of candidates is a finite number, so V(P)=$\{x\in A|lp(x) \textrm{ is minimum} \}\ne \emptyset$, which means LU satisfies AAW criterion.

Assume that alternative $a$ appears at the bottom of at least half of the individual preference lists. Since $|A|\ge 3$, so $a\not\in\{x\in A|lp(x) \textrm{ is minimum} \}$, $a\not\in V(P)$, which means LU satisfies SDC.

LU satisfies monotonicity because raising
candidate $a$ up on some preference lists can never increase the number
of last-place votes that $a$ receives, nor can it decrease the number of
last-place votes that any other candidate receives. Hence, if $a$ wins the
election before such a change, she wins afterwards as well.
\end{proof}

The Pareto criterion is important in the context of Arrow's impossibility theorem \cite{A1,arrow1}. However, LU does not satisfy this criterion. Suppose we are seeking a voting procedure that satisfies SDC, Pareto, and monotonicity simultaneously.
Here, we introduce a new voting protocol that satisfies Pareto and monotonicity, and also prevents social disappointment in voting systems.

\begin{definition}(\textit{Least Unpopular Reselection (LUR)})
First, the set of alternatives appearing least often at the bottom of individual preference lists (i.e., the set of least unpopular alternatives) is chosen. If this set has only one member, it is the social choice. Otherwise, the remaining alternatives (if any) are removed and the procedure LU is run for the set obtained from the previous stage. This procedure is repeated until it cannot be continued (because a new set of alternatives cannot be produced). The set obtained in the last repetition is the set of social choice. Formally, one repeatedly removes the alternative or the alternatives except those with the least last-place votes, with the last group of alternatives to be
removed was tied for the win. More precisely, V is the least unpopular reselection (LUR) if V = V$^*_{LUR}$, where V$_{LUR}$(P) is the set of all alternatives
with the least last-place votes in P (and all tie if all have the same
number of first-place votes).\label{lur}
\end{definition}

To clarify the definitions of LU and LUR, we illustrate these two voting procedures with a single example.

\begin{example}\label{ex2}
Consider the following voting profile:
\begin{table}[ht]
\centering
\begin{tabular}{ccc}
\toprule
Voter 1 & Voter 2 & Voters 3 and 4  \\
\midrule
a & a & b \\
c & b & c \\
b & c & a \\
\bottomrule
\end{tabular}
\end{table}
For each of LU and LUR procedures, we calculate what the resulting social choice is.

\noindent LU: Since $lp(a)=2, lp(b)=1$, and $lp(c)=1$ then $\{b,c\}$ is the the social choice when the LU procedure is used.

\noindent LUR: We decide which alternative(s) occur(s) at the bottom of the fewest lists and remove the remaining alternative(s) from all the lists. In this phase, $a$ is removed from each list leaving the following:
\begin{table}[ht]
\centering
\begin{tabular}{ccc}
\toprule
Voter 1 & Voter 2 & Voters 3 and 4  \\
\midrule
c & b & b \\
b & c & c \\
\bottomrule
\end{tabular}
\end{table}

Now, $c$ occurs at the bottom of 3 of 4 lists, and thus is
eliminated. Hence, $b$ is the social choice when the LUR procedure is used.
\end{example}

\begin{proposition}
LUR protocol satisfies the AAW, Monotonicity, Pareto criteria, and SDC, but does not satisfy the CWC, CLC, and IIA criteria.
\end{proposition}
\begin{proof}
If candidate $a$ is
ahead of candidate $b$ on every preference list, candidate $a$ has no last-place votes and therefore $b$ suffers elimination in some round. Hence the
LUR method satisfies Pareto.

The rest of proof is the same as proof of Proposition \ref{lu1} and \ref{lu2}.
\end{proof}

\subsection{Properties of our Voting Procedures}

Some of results from this paper are summarized in Table \ref{summary}. This table provides the answers to 63 questions of the form
"Does method X satisfy criterion Y ?" The rows are indexed by 9 methods. The columns are indexed by 7 criteria: always-a-winner, the Condorcet winner, Pareto, monotonicity, independence of irrelevant alternatives, social disappointment, and Condorcet loser criterion.

Based on the comparison in Table \ref{summary}, first, no method satisfies all 7 of the criteria. This should not surprise us, because Taylor shows that no method can simultaneously satisfy AAW, IIA, and CWC \cite{Taylor3}. Even if we are willing to dispense with IIA, we still face a difficulty because Theorem \ref{thm:impossibility} in the next section shows that no method (for four and more candidates) can satisfy CWC and SDC simultaneously. It should be noted that on of the advantage of LU and LUR methods in compare with Coombs is that LU and LUR satisfy monotonicity criterion but Coombs method does not. 

\begin{table}[ht]
\caption{Comparison of voting procedures.}\label{summary}
\centering
\begin{tabular}{|l|c|c|c|c|c|c|c|}
\hline
 & AAW & CWC & Pareto & Mono& IIA & SDC&CLC\\
\hline
Condorcet & NO & YES & YES & YES & YES & NO & YES\\
\hline
Plurality & YES & NO & YES & YES & NO & NO & NO\\
\hline
Borda & YES & NO & YES & YES & NO & NO & YES\\
\hline
Hare & YES & NO & YES & NO & NO & NO & NO\\
\hline
Seq. Pairs & YES & YES & NO & YES & NO & NO & YES\\
\hline

Copeland & YES & YES & YES & YES & NO & NO & YES \\
\hline
\rowcolor{yellow}
 Coombs & YES & NO & YES & NO & NO & YES & NO\\
\hline
\rowcolor{blue!25}
LU & YES & NO & NO & YES & NO & YES & NO\\
\hline
\rowcolor{orange!50}
LUR & YES & NO & YES & YES & NO & YES & NO \\
\hline
\end{tabular}
\end{table}

\section{Choice of Condorcet Winners Under the Single-peaked Preferences Assumption}
Supporters of AV (such as Donald Horowitz) argue that, compared to plurality, AV has a greater likelihood of the end result being a victory by a `moderate' candidate (i.e. one who enjoys broad support but is not necessarily the
first place choice of many voters) as opposed to an `extremist' candidate (i.e.
one who may have first place support from a substantial number of voters but
who would lose to one or more of the other candidates in the contest if there
were to be a two-candidate head on head contest). In social choice terminology
\cite{Black,Saari}, we may translate this into the claim that AV has a
higher probability of choosing the Condorcet winner (when one exists) than does
plurality. The basic argument is simply that, under AV, unlike what is true for
plurality, preferences for moderate candidates who are Condorcet winners, but
who may not be given much first place support, have at least the potential to be
important in deciding election outcomes by virtue of the sequential nature of the
ballot transfer process \cite{GF}. 

Grofman and Feld \cite{GF} showed that under the
single-peaked preferences assumption, for four alternatives or fewer, AV is always as likely or more likely to
select the Condorcet winner than plurality. However, under the
same assumptions, the Coombs rule will always select the Condorcet winner regardless of
the number of alternatives. For the voting procedure LUR we prove the following proposition.
\begin{proposition}
Under the assumptions of single-peaked preferences and no party holding a majority of first place preferences, when we have four parties or fewer, a candidate of the median party is always being selected when voting is conducted under the LUR procedure, and this property does not hold in general when voting is conducted under first-past-the-post\footnote{A first-past-the-post (FPTP) voting system is one in which voters indicate on a ballot the candidate of their choice, and the candidate who receives the most votes wins.}. 
\end{proposition}
\begin{proof}
The proof of this proposition for two alternatives is trivial, so we
focus on the case of three and four alternatives.

First, we show that under the simplifying assumptions of three (or four) parties and
single-peaked preferences, the median party (the majority winner in pairwise contest) must be a moderate party. Our proof will be by contradiction. Without loss of generality imagine that the leftmost party is the median party. By this hypothesis the leftmost party (since it is supported by the median voter) receives a majority of first place votes. However, we have posited that no party has a majority of first place preferences. This contradiction demonstrates that, for single-peaked preferences among three (or four) alternatives, when there is no party with a majority of first-place preferences, no extreme party can be the median party. This means that there are alternatives on both sides of the alternative supported by the
median voter, then that alternative is no voter's last choice. So, the only last choices are the extremes at the two ends of the dimension. So, the Condorcet winner (i.e., the alternative favored by the median voter) cannot be the alternative with the most last place votes. 
\begin{itemize}
    \item[Case I.] With three alternatives, the the alternative favored by the median voter (which is none of the extremes at the two ends of the dimension) is the alternative with the least last place votes i.e., zero. So, using the LUR procedure leads to the electing of the Condorcet winner in this case. However, as the following example shows, there are situations that the outcome of the election by the FTPT procedure is not the Condorcet winner in this case. 
    \begin{table}[ht]
\centering
\begin{tabular}{ccc}
\toprule
Voters 1 and 2 & Voter 3 & Voters 4 and 5  \\
\midrule
a & b & c \\
b & a & b \\
c & c & a \\
\bottomrule
\end{tabular}
\end{table}

This profile satisfies the assumptions, because the preferences are single-peaked and and no party holding a majority of first place preferences. The following graphs show three preferences that are single-peaked over outcomes $\{a,b,c\}$. On the vertical axis, the number represents the preference ranking of the outcome, with 1 being most preferred. 
\begin{center}
    \begin{tikzpicture}[domain=0:2]
    \draw[ultra thin,color=gray] (-0.1,-0.1) grid (2,2);
    \draw[->] (-0.1,0) -- (2.2,0);
    \draw[->] (0,-.1) -- (0,2.2);
    \draw	(0,0) node[anchor=north] {a}
		(1,0) node[anchor=north] {b}
		(2,0) node[anchor=north] {c}
		(0,0) node[anchor=east] {3}
		(0,1) node[anchor=east] {2}
		(0,2) node[anchor=east] {1};
    \draw[thick,color=red] plot[id=x] function{x};
    \draw[thick,color=blue] (0,1) -- (1,2) -- (2,0);
    \draw[thick,color=orange] plot[id=exp] function{2-x};
\end{tikzpicture}
\end{center}

'b' is the Condorcet winner and the choice of the median voter. However, the outcome of the FPTP procedure is not 'b'.

    \item[Case II.] Since the Condorcet winner (i.e., the alternative favored by the median voter) cannot be the alternative with the most last place votes, using the LUR procedure leads to removing the extremes at the two ends of the dimension. So, the Condorcet winner is never eliminated at the first round of the LUR procedure and therefore will be declared a winner at the second round in which there are only two alternatives left, one of whom is the alternative preferred by the median voter, since that alternative is the Condorcet winner (and thus defeats each and every other alternative in paired competition), it will necessarily be chosen. However, as the following example shows, there are situations that the outcome of the election by the FTPT procedure is not the Condorcet winner in this case. 
    \begin{table}[ht]
\centering
\begin{tabular}{cccc}
\toprule
Voters 1-3 & Voters 4 and 5 & Voter 6 &Voters 7-9  \\
\midrule
a & b & c & d\\
b & a & b & b\\
c & c & a & c\\
d & d & d & a\\
\bottomrule
\end{tabular}
\end{table}

This profile satisfies the assumptions, because the preferences are single-peaked and and no party holding a majority of first place preferences. The following graphs show three preferences that are single-peaked over outcomes $\{a,b,c,d\}$. On the vertical axis, the number represents the preference ranking of the outcome, with 1 being most preferred. 
\begin{center}
    \begin{tikzpicture}[domain=0:3]
    \draw[ultra thin,color=gray] (-0.1,-0.1) grid (3,3);
    \draw[->] (-0.1,0) -- (3.2,0);
    \draw[->] (0,-.1) -- (0,3.2);
    \draw	(0,0) node[anchor=north] {a}
		(1,0) node[anchor=north] {b}
		(2,0) node[anchor=north] {c}
		(3,0) node[anchor=north] {d}
		(0,0) node[anchor=east] {4}
		(0,1) node[anchor=east] {3}
		(0,2) node[anchor=east] {2}
		(0,3) node[anchor=east] {1};
    \draw[thick,color=red] plot[id=x] function{x};
    \draw[thick,color=blue] (0,2) -- (1,3) -- (3,0);
    \draw[thick,color=green] (0,1) -- (2,3) -- (3,0);
    \draw[thick,color=orange] plot[id=exp] function{3-x};
\end{tikzpicture}
\end{center}

'b' is the Condorcet winner and the choice of the median voter. However, the outcome of the FPTP procedure is not 'b'.
\end{itemize}
\end{proof}

\section{An Impossibility Theorem Based on the Concept of Social Disappointment in Voting Systems}\label{impossibility}
In this section, we provide a version of the impossibility theorem regarding the notion of social disappointment in elections, showing that there is no voting rule for four or more candidates that simultaneously satisfies SDC and CWC.

Arrow's theorem in \cite{A1}, have motivated a vast amount of research on multicandidate elections. There
are now several dozen Arrow-type impossibility theorems that address a wide array of
social choice situations, but all have the same theme of the collective incompatibility of
conditions which, taken separately, seem reasonable and appealing \cite{bramsfishburn}.

Taylor proved in \cite{Taylor3} and also \cite[pp. 28-31]{Taylor}  that there is no voting procedure for three or more alternatives that satisfies the always-a-winner criterion, the independence of irrelevant alternatives, and the Condorcet winner criterion. What follows is an impossibility theorem based on the concept of social disappointment in voting systems. This theorem can be seen as part of the story of the difficulty with
"reflecting the will of the people." The proof of this theorem, like that of Arrow's theorem, makes critical use of the voting paradox of Condorcet \cite{Condorcet}.

\begin{theorem} \label{thm:impossibility}
There is no voting procedure for \textbf{four or more} alternatives that satisfies the SDC and the Condorcet winner criterion.
\end{theorem}
\begin{proof} 
Assume that there is a voting procedure that satisfies the Condorcet winner criterion. We show that if this procedure is applied to the profile that consists of the Condorcet's voting paradox \cite{Condorcet}, then it produces a winner which will lead to social disappointment.

Suppose that there is a voting procedure that satisfies the Condorcet  winner criterion. Consider the following profile for (n+1) alternatives and 2n voters, where $n\ge 3$ (each column corresponds to a ballot):
\begin{equation*}
 \begin{array}{llllllllll}
x_{n+1} & x_{n+1} & \cdots & x_{n+1} & x_{n+1} & x_1& x_2&\cdots&x_{n-1}&x_n\\
x_1& x_2&\cdots&x_{n-1}&x_n &x_2&x_3&\cdots&x_n & x_1 \\
x_2&x_3&\cdots&x_n & x_1 & x_3&x_4&\cdots&x_1& x_2\\
\vdots&\vdots&\ddots&\vdots& \vdots&\vdots&\vdots&\ddots&\vdots& \vdots \\
x_{n-1} & x_n & \cdots & x_{n-3} & x_{n-2} & x_n&x_1&\cdots&x_{n-2}&x_{n-1} \\
x_n&x_1&\cdots&x_{n-2}&x_{n-1} & x_{n+1} & x_{n+1} & \cdots & x_{n+1} & x_{n+1} \\
\end{array}   
\end{equation*}

The alternative $x_{n+1}$ is the unique social choice when the Condorcet's method is used. Although the alternative $x_{n+1}$ is a social choice, it is at the bottom of half of the individual preference lists and so social disappointment has taken place. 
\end{proof}

\begin{proposition}
There is a voting procedure for \textbf{three} alternatives that satisfies the SDC and the Condorcet winner criterion.
\end{proposition}
\begin{proof}
None of the voting procedures mentioned in this article satisfies the SDC and the Condorcet winner criterion. Here we design a hybrid voting rule that we call \textit{unique-Condorcet Coombs} (UCC) method and show that this protocol satisfies the SDC and the Condorcet winner criterion. UCC is the method that selects the Condorcet candidate when it is unique, otherwise selects the output of the Coombs method. 

If in the Condorcet method more than half of voters place x at the bottom
of individual preference lists then for sure x would not be a social choice with UCC method, and
in this case social disappointment would not occur. However, if the number of voters is an
even number and precisely half of the voters place x at the end of their lists, one of these two
possibilities will happen:
\begin{itemize}
\item [(1)] Not all the voters in the other half place x at the top of their lists, in which 
case, x will definitely not be a social choice with the UCC method, and social disappointment will not occur.
\item [(2)] All the voters in the other half also place x at the top of their lists, in which case x will definitely be in the set of social choice with Condorcet method. In this case, if there are only three alternatives, the set of social choice with Condorcet method will certainly have more than one member, and so the UCC procedure uses Coombs method to determine the output of the election. Therefore, SD will not occur. 
\end{itemize}
As a result, UCC method satisfies the SDC and the Condorcet winner criterion when we have three alternatives in the election.
\end{proof}

\section{Evaluation} 
In section \ref{SFSD}, we showed Coombs, LU, and LUR voting protocols avoid social disappointment. In this section, we empirically evaluate the performance of discussed voting protocols from two point of views: \emph{avoidance  of social disappointment} and \emph{resistance to manipulation} via strategic voting because it is an important benchmark for desirability of a voting procedure \cite{conitzerwalsh}. For this purpose, we designed two different setups and we conducted a sensitivity analysis to evaluate the performance of voting rules in both cases. 
\subsection{Experimental Setup}
In this section, we first define the independent (IV-SD) and dependent variables
(DV-SD) for performance evaluation of voting procedures relating to SDC and resistance against manipulation.

\subsubsection{Social disappointment experimental setup}
\textbf{IV-SD-1}: \textit{number of candidates}. To identify how vulnerable
    each voting rule is with respect to SDC,
    we vary the number of candidates between 3 and 6. For two reasons we did not consider the number of candidates greater than 6. First, in many real-world scenarios like presidential elections in many countries in the world \footnote{\url{https://libguides.princeton.edu/elections/foreign}}, variants of voting rules are used, and most of the time the number of candidates is relatively small. Second, based on our experiments, the chance of social disappointment occurrence in elections decreases when the number of candidates increases, and so we did not consider a number of candidates greater than 6.
    
\noindent\textbf{IV-SD-2}: \textit{number of voters}. Voting protocols may perform differently depending on how many voters have participated in the election. So, we  generate profiles with different number of voters between 6 and 10. Based on our experiments, the chance of occurrence of SD decreases in elections,  where the number of candidates is fixed but the number of voters increases \cite{annonymous}. This is in particular relevant for considerable number of situations such as multiagent systems, where the number of voters are relatively small.

\noindent\textbf{IV-SD-3}: \textit{number of profiles}. Since a well-designed experiment takes into account the randomness of elections, we generate 1000 profiles for each pair of $i$ and $j$, where $i$ is the number of candidates and $j$ is the number of voters.
    
    \noindent\textbf{IV-SD-4}: \textit{voting procedures}. We compare the performance of nine voting rules i.e., plurality,Condorcet, Borda, Hare, Coombs, Copeland, Seq. Pairs., LU, and LUR with respect to SD.
    
    \noindent\textbf{DV-SD-1}: \textit{outcome of elections}. To assess  the vulnerability of voting rules concerning SD, we check the occurrence of SD in elections for each random profile and voting procedure.

In summary, in order to test the vulnerability of voting rules to electing a candidate who violates social disappointment criterion, we generated 1000 random profiles for each pair of $i$ and $j$, where $i=3,4,5,6$ is the number of candidates and $j=6,7,8,9,10$ is the number of voters. For this purpose, we used \texttt{sample()} function in R language to generate a random ballot based on the number of candidates for each voter in a profile. For simplicity, we assumed that the distribution of the society is uniform and all of ballots are valid. Then, we run nine mentioned voting procedures in this paper i.e., plurality, Condorcet, Borda, Hare, Coombs, Copeland, Seq. Pairs., LU, and LUR on each profile. Finally, we checked whether SD for each election has happened or not. An R language package that implements our algorithms is available at \cite{annonymous}.

\subsubsection{Manipulation experimental setup}
Since an important benchmark for desirability of a voting rule is its resistance to manipulability via strategic voting, we designed four scenarios of manipulating elections such as control and bribery, and tested the performance of each voting protocol in all different scenarios.  Before explaining the details of setup for each scenario, first, we briefly review the meaning of manipulation, control, and bribery in elections.

In many situations, voters may vote strategically. That is, they may declare preferences
that are not their true ones, with the aim of obtaining a better outcome for themselves. This is often referred to as \textit{manipulation} or \textit{strategic voting}. Voting rules that are never manipulable are also referred to as
\textit{strategyproof} \cite{conitzerwalsh}.  The Gibbard-Satterthwaite theorem shows that there is no strategy-proof
voting rule that simultaneously satisfies certain combinations of desirable properties \cite{Gibbard,Satterthwaite}. Control and bribery are two families of problems modeling various
ways of manipulating elections. Control problems model situations where
some entity, usually referred to as the chair or the election organizer, has some ability
to affect the election structure by adding or deleting voters or candidates. On the other hand, bribery models situations where the structure of the election stays intact (we have the same candidates and the same voters), but some outside agent pays the voters to change their votes \cite{FaliszewskiRothe,EHH}. 

Now, we define the independent (IV-M) and dependent variables (DV-M) for performance evaluation of voting procedures against different scenarios of manipulation.

\textbf{IV-M-1}: \textit{number of candidates}. To identify how vulnerable each voting rule is with respect to manipulation, we vary the number of candidates between 3 and 10. In many real-world situations\footnote{\url{https://libguides.princeton.edu/elections/foreign}}, variants of voting rules is used and most of the time the number of candidates are relatively small and less than 10. 

\textbf{IV-M-2}: \textit{number of voters}. In order to create a more realistic profiles, we  generate random profiles with different number of voters i.e., 10, 100, and 1000.

\textbf{IV-M-3}: \textit{number of profiles}. 
Since a well-designed experiment takes into account the randomness of elections, we generate 30 profiles for each pair of $i$ and $j$, where $i$ is the number of candidates and $j$ is the number of voters. 

\textbf{IV-M-4}: \textit{manipulation scenarios}. We consider 4 manipulation scenarios to evaluate the robustness of voting procedures against manipulation (we elaborate each of them later in this section):
\begin{enumerate}
    \item Constructive Control by Adding/Deleting Voters' Ballots,
    \item Constructive Control by Adding/Deleting Candidates,
    \item Bribery/Self-manipulation, and
    \item Social Network and Social Media Influence on Voters' Preferences During Election Days.    
\end{enumerate}

\textbf{IV-M-5}: \textit{voting procedures}.
We compare the performance of nine voting rules i.e., plurality, Condorcet, Borda, Hare, Coombs, Copeland, Seq. Pairs., LU, and LUR with respect to manipulation scenarios.

\textbf{DV-M-1}: \textit{outcome of elections after manipulation}. We compare the outcome of original elections with the outcome of elections after each manipulation scenario for a large number of configurations.   

In summary, to compare the performance of LU and LUR with other seven well-known voting rules against manipulation, we  generated 30 random profiles for each pair of $i$ and $j$, where $i=3,4,\dots,10$ is the number of candidates and $j=10,100,1000$ is the number of voters. For this purpose, we used the \texttt{sample()} function in R language to generate a random ballot based on the number of candidates for each voter in a profile, we assumed that the distribution of the society is uniform and all of ballots are valid. Then, we run all above mentioned voting procedures on each profile. 

The manipulation scenarios we considered for the independent variable \textbf{IV-M-4} are based on the following justification:

\begin{enumerate}
    \item \textbf{Constructive Control by Adding/Deleting Voters' Ballots}: The issue of control by adding, deleting voters ballots (or a combination of them) is very natural and, indeed, happens in real-life political elections. For example, it is widely speculated that "adding" about 8 millions ballots in favor of
    Mahmoud Ahmadinejad to the 2009 Islamic Republic of Iran presidential election had the effect of ensuring Ahmadinejad's victory
    (otherwise, Mir-Hossein Mousavi would have won or gone to the second round i.e., the runoff election)\footnote{\url{http://www.bbc.com/persian/46110885}}. We consider a situation in which the election organizer (in this example, the chair of the electoral commission) is able to delete randomly 10 percent of voters' ballots  and replace all of them by his or her individual preference list. For this case, we removed randomly 10 (or 20) percent of voters' ballots  and then replaced all of them with a fixed random ballot.
    \item \textbf{Constructive Control by Adding/Deleting Candidates}: One control action the chair might exert is to change the candidate set, either by
    adding some new candidates from a given set of spoiler candidates (hoping to make
    $p$'s most competitive rivals weaker relative to $p$), or to delete up to $k$ candidates from
    the given election (to get rid of $p$'s worst rivals) \cite{FaliszewskiRothe}. For example, it is widely speculated that "adding"
    Nader to the 2000 U.S. presidential election had the effect of ensuring Bush's victory
    (otherwise, Gore would have won) \cite{FaliszewskiRothe}. We consider a situation in which the chair's goal in exerting some control action is to
    make a given candidate $p$ the winner of the resulting election. In this scenario, candidate $a$ performs very well in pre-election polls and candidate $b$ is the most competitive rival of $p$. The chair deletes candidate $b$ in order to make $p$'s most competitive rival (i.e., candidate $a$)  weaker relative to $p$. We removed the candidate who has obtained third rank in the election without manipulation (one can consider this as the result of pre-election polls) to increase the chance of the candidate, who has earned the second rank in the pre-election polls, to being a winner.
    \item \textbf{Bribery/Self-manipulation}: Election bribery problems model situations where an outside agent wants a particular alternative to win and pays
    some of the voters to vote as the agent likes \cite{MRS,Yong,FaliszewskiRothe,KaczmarczykFaliszewski}. Self-manipulation problem addresses situations where the supporters of candidate $p$ know that they have no chance to win the election; however, most of them agree that candidate $q$-who performs well in pre-election polls-is the closest person to their political ideals among other candidates. So, this group of voters forms a coalition and manipulates the election by voting non-truthfully, i.e., by swapping between their sincere first choice candidate $p$ and their second choice candidate $q$ who has a better chance of winning the election. Assume that other voters vote truthfully. Here, we consider a situation in which there is a minority group of voters who decide that their second choice candidate has a better chance of winning the election, and so the members of the group swap their sincere first choice with their second choice. For this case, according to the result of the election without manipulation (indicated by the result of pre-election polls) there is a minority group of voters who decides that its second choice candidate has a better chance of winning the election, so it swaps its sincere first choice with its second choice.
    \item \textbf{Social Network and Social Media Influence on Voters' Preference Lists}: As online social networks have become significant sources of information for potential voters, a new tool in an attacker's arsenal is to effect control by harnessing social influence, for example, by spreading fake news and other forms of misinformation through online social media \cite{WilderVorobeychik}. Also, it is possible that the sincere preferences of the voters are influenced by the votes of their friends. In this case, each
    agent votes strategically, taking into consideration both her preferences, and her (limited) information about the preferences of other voters, assuming that the information the agent has comes from her friends in the social network and from a public opinion poll \cite{SinaHazonHassidim}. In real life, political parties often try to influence the outcome of elections by means of social media (e.g., see \cite{BessiFerrara,KollanyiHowardWoolley}). For example, social bots (or chatbots)\footnote{The chatbots are basic software programs with a bit of artificial intelligence and rudimentary communication skills. They can send messages on Twitter based on a topic, usually defined on the social network by a hashtag symbol, like \#Clinton.} distorted the 2016 U.S. presidential election online discussion. According to \cite{BessiFerrara}, the presence of social media bots can indeed negatively affect democratic political discussion rather than improving it, which in turn can potentially alter public opinion and endanger the integrity of the presidential election.
    Also, authors of \cite{WilderVorobeychik} show that election control through social influence is a salient threat to election integrity. 
    Here, we consider a situation that 10 percent of voters exchange their sincere first choice candidate with candidate $p$ under influence of $p$'s advertisements, fake news, rumors, and chatbots through social media during election days. For this case, we randomly selected 10 percent of voters and changed their ballots in a way that  the candidate, who has earned the second rank in the  election without manipulation, became their first choice (without changing their preferences about other candidates).  
\end{enumerate}

In the next step, we run our voting rules on each new profile after manipulation. Finally, we checked that whether manipulation had effect on the outcome of each election or not.
In the next subsection, we show that voting procedures LU, LUR, and Coombs had better performance against manipulation, in all of above scenarios, than other seven voting rules in this paper.

\subsection{Experimental Results for Social Disappointment in Voting Procedures}
Figure \ref{fig:sd} shows the performance of voting protocols against social disappointment.  As we expected (see section \ref{SFSD}), social disappointment does not happen for Coombs, LU, and LUR methods. However, other procedures cannot prevent social disappointment. Among voting procedures, Plurality has the worst performance. An interesting observation is that the number of social disappointment occurrence in elections decreases when the number of candidates increases as one can see in Figure \ref{fig:sd}. Another noticeable point is that Borda and Copeland violate SDC in fairly small number of cases, indicating that  social disappointment happens for these methods just in rare cases (see the proof of the Proposition \ref{lu1} that supports this observation theoretically). 
\begin{figure}[ht]
	\centering
	\includegraphics[scale=.3]{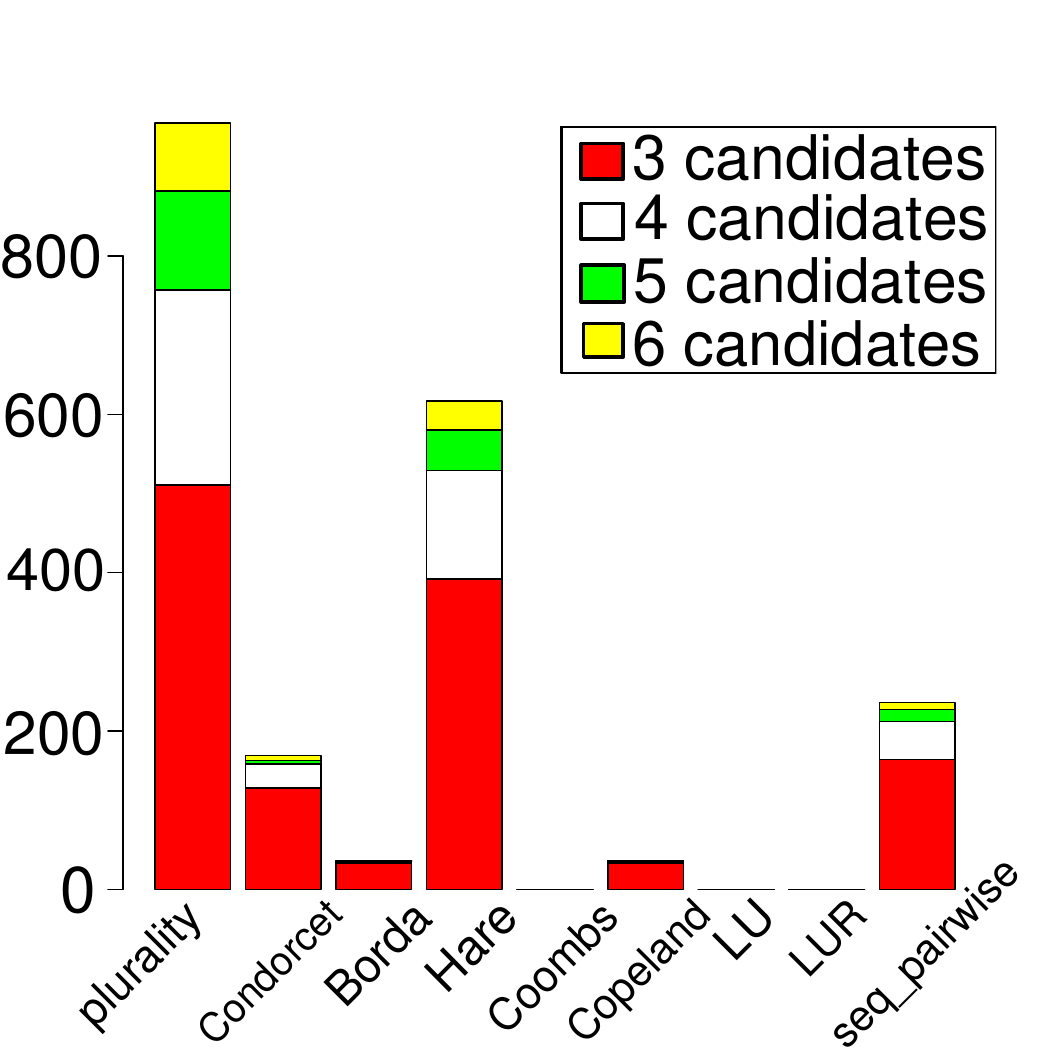}
	\caption{Performances of voting procedures: Number of occurrence of social disappointment in different elections has been shown with different colors based on the number of candidates. Social disappointment does not happen for Coombs, LU, and LUR. 
	When the number of candidates increases, the number of social disappointment occurrence in elections decreases. The performance of Borda and Copeland rules are acceptable in this regard.
	} \label{fig:sd}
\end{figure}

\subsection{Experimental Results for Resistance Against Manipulation in Voting}
In this section, we compare the outcome of original elections with the outcome of elections after each manipulation scenario for each scenario explained in experimental setup:
\begin{enumerate}
    \item \textbf{Constructive Control by Adding/Deleting Voters' Ballots}: In this case, as shown in Figure \ref{fig:10percent}, LU and LUR are more robust against manipulation in this scenario compared to other procedures. Figure~\ref{fig:10percent} shows that the number of affected elections in this scenario is independent of the number of candidates. Except for LU and LUR, other procedures perform as bad as plurality rule in this scenario.
    \begin{figure}[ht]
	\centering
	\includegraphics[scale=.3]{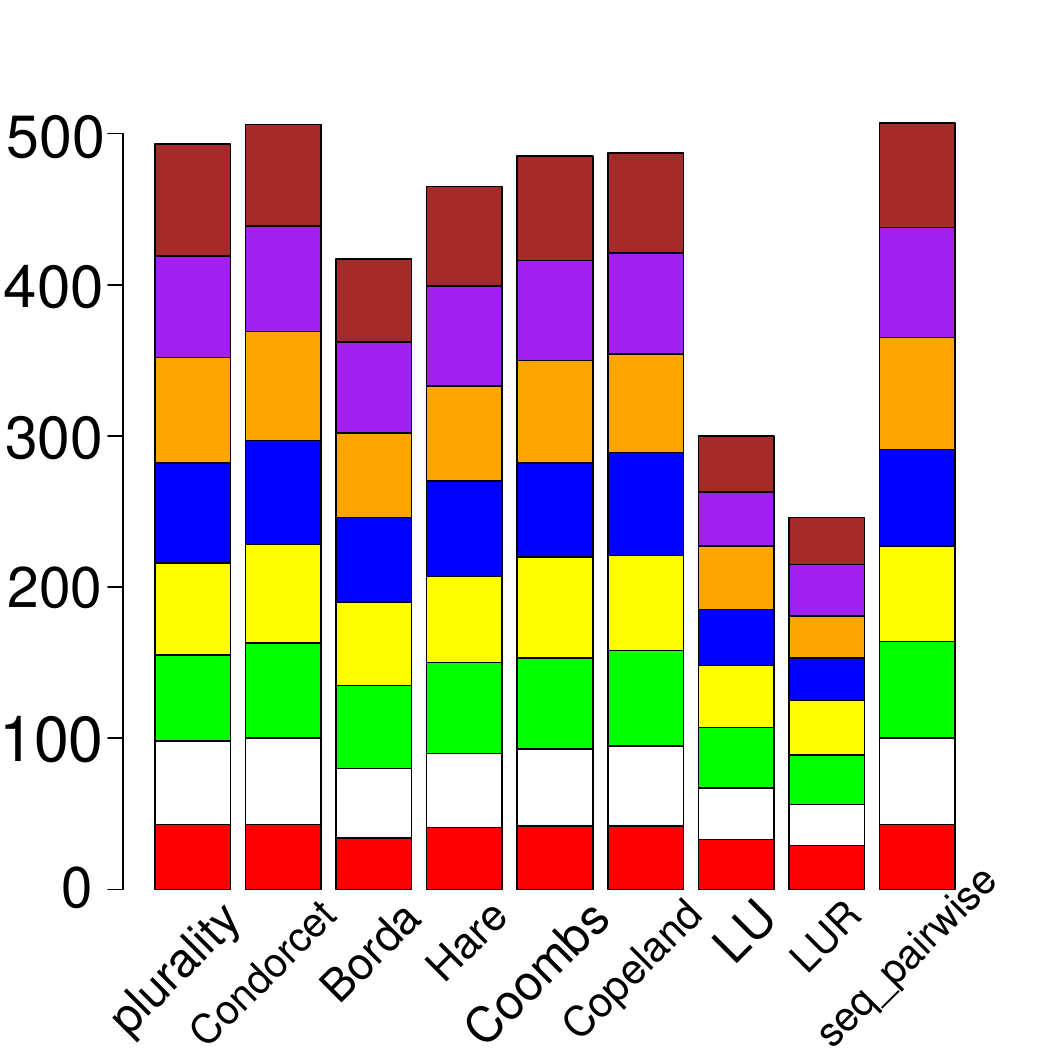}
	\caption{Performances of voting procedures against manipulation in constructive control by deleting and replacing 10 percent of voters' ballots. LU and LUR are more robust than others, and a meaningful difference regarding the performance of other methods cannot be seen  in this scenario. 
	The number of affected elections is shown with different colors corresponding to the number of candidates.
	} \label{fig:10percent}
	\end{figure}
	\item \textbf{Constructive Control by Adding/Deleting Voters' Ballots} (second scenario): As shown in Figure \ref{fig:20percent}, LU and LUR are more robust against manipulation in this constructive control manipulation scenario compared to other procedures. No meaningful difference can be seen relating to the performance of other seven protocols in this scenario. Figure \ref{fig:20percent} shows that the number of affected elections in this scenario is independent of the number of candidates.
    \begin{figure}[ht]
	\centering
	\includegraphics[scale=.3]{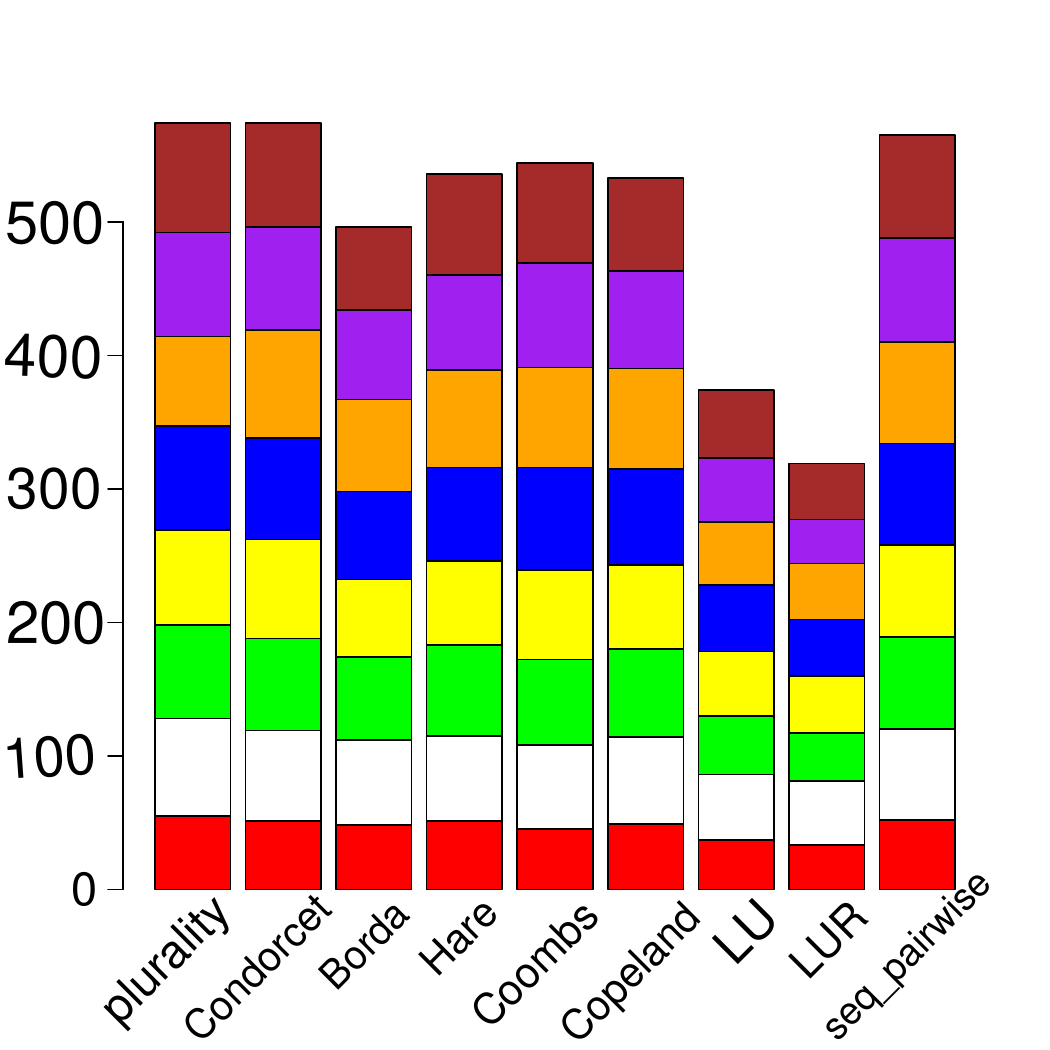}
	\caption{Performances of voting procedures against manipulation in the scenario of constructive control by deleting and replacing 20 percent of voters' ballots. LU and LUR present better performance than other methods, and also other procedures perform as bad as plurality rule in this scenario. 
	The number of affected elections is shown with different colors corresponding to the number of candidates.
	} \label{fig:20percent}
    \end{figure}
    \item \textbf{Constructive Control by Adding/Deleting Candidates}: As shown in Figure \ref{fig:quit}, Coombs, LU and LUR are more robust against manipulation compared to other procedures. 
    In contrast with the first scenario i.e., constructive control by adding/deleting voters' ballots, Copeland performs slightly better than Borda in this case. Also, in the evaluated scenarios, the number of affected elections in this scenario does not depend on the number of candidates. 
    \begin{figure}[ht]
	\centering
	\includegraphics[scale=.3]{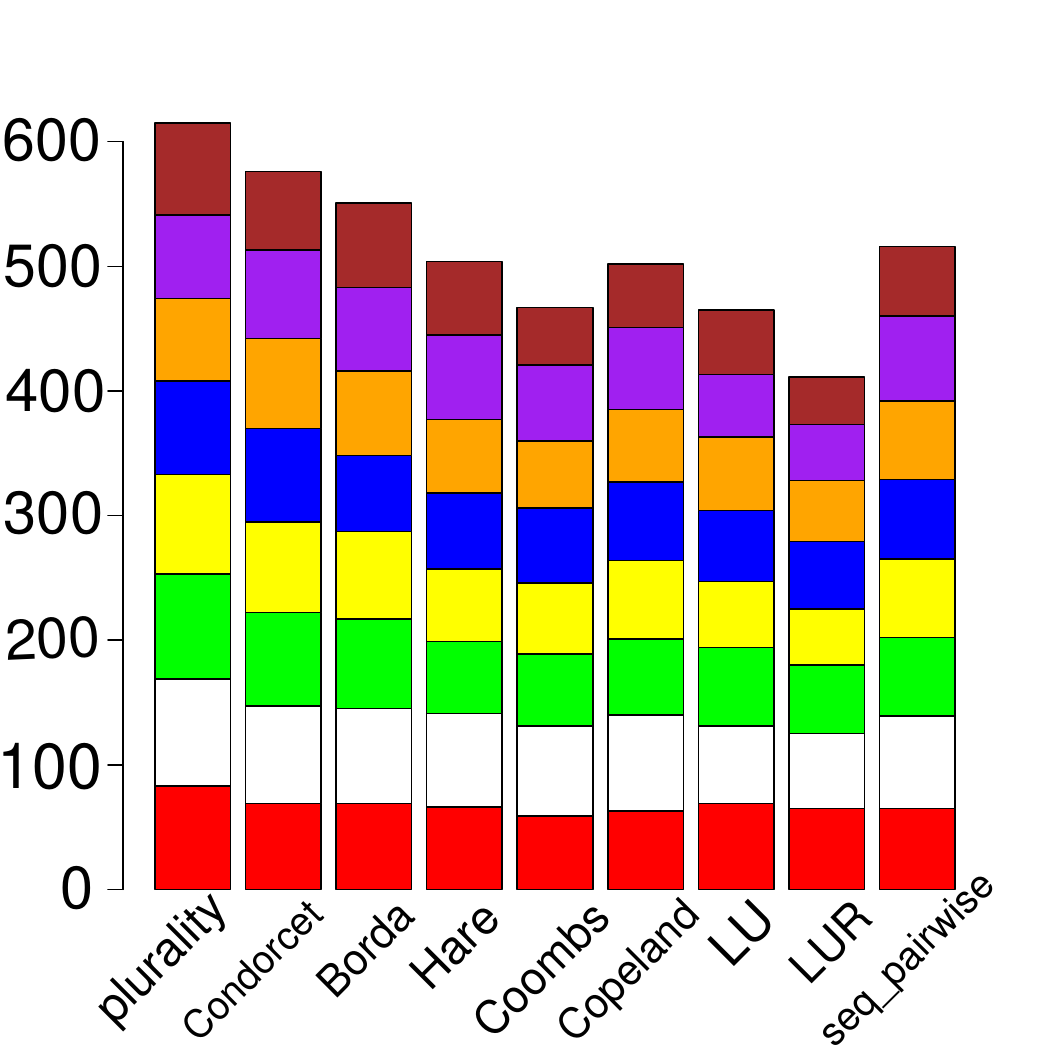}
	\caption{Performances of voting procedures against manipulation in constructive control by deleting a candidate. Coombs, LU, and LUR have the best performance and plurality has the worst. 
	Number of affected elections has been shown with different colors corresponding to the number of candidates.
	} \label{fig:quit}
    \end{figure}
    \item \textbf{Bribery/Self-manipulation}:  As shown in Figure \ref{fig:bribery}, LU and LUR are more robust against manipulation in this scenario compared to other procedures. 
    Except for plurality and Borda rules, the other seven well-known methods have acceptable performance in this scenario. 
    In contrast with first and last scenario, Borda surprisingly does not perform as good as Copeland, Condorcet, Seq. Pairs., Hare, and Coombs. 
    \begin{figure}[ht]
	\centering
	\includegraphics[scale=.3]{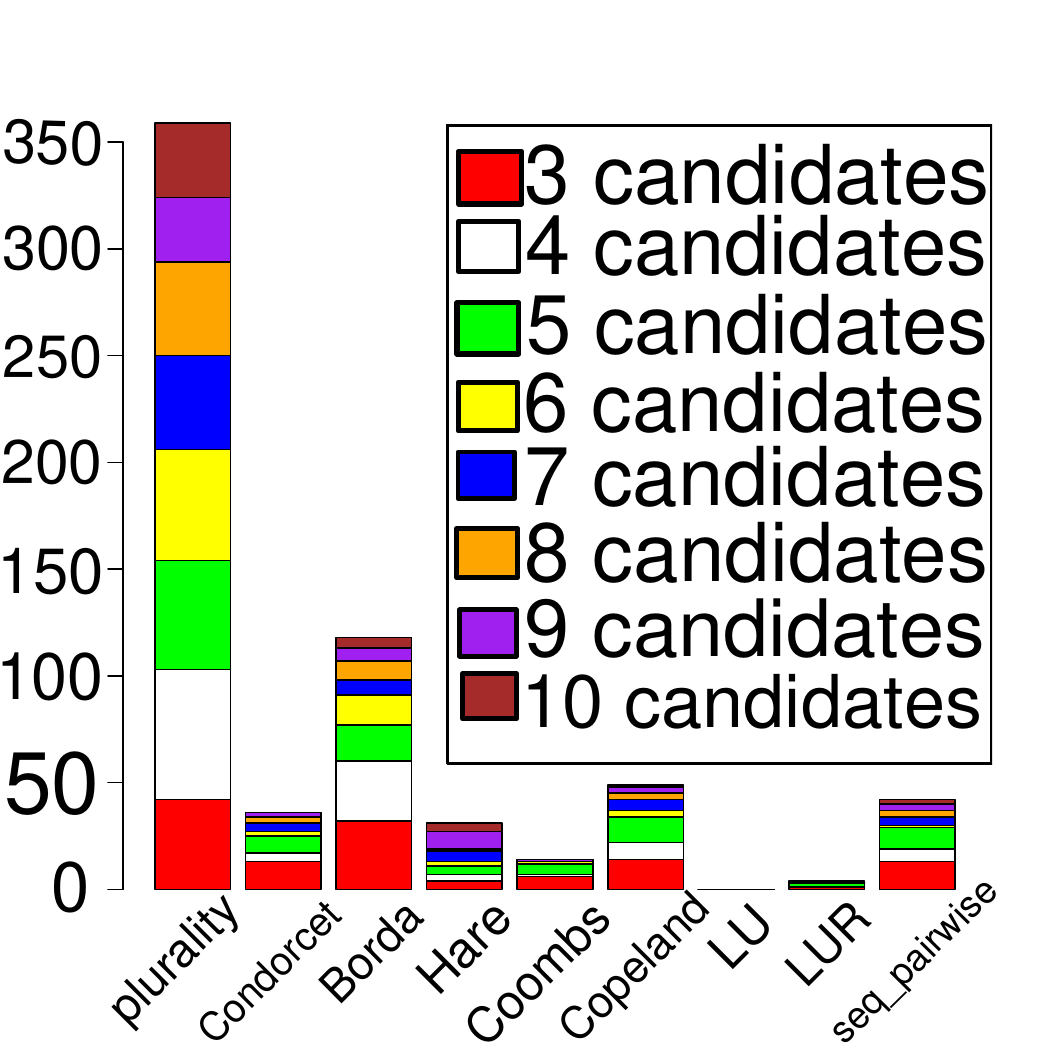}
	\caption{Performances of voting procedures against manipulation in bribery/self-manipulation. The best performance belong to LU and LUR. Except for plurality and Borda rules, other seven methods have acceptable performance. 
	The number of affected elections is shown with different colors corresponding to the number of candidates.
	} \label{fig:bribery}
    \end{figure}
    \item \textbf{Social Network and Social Media Influence on Voters' Preference Lists}:  As shown in Figure \ref{fig:influence}, LU and LUR are more robust against manipulation. The worst performance belongs to plurality, and the Coombs method has the best performance after LU and LUR. Also, the number of affected elections in this scenario is independent of the number of candidates. 
    \begin{figure}[ht]
	\centering
	\includegraphics[scale=.3]{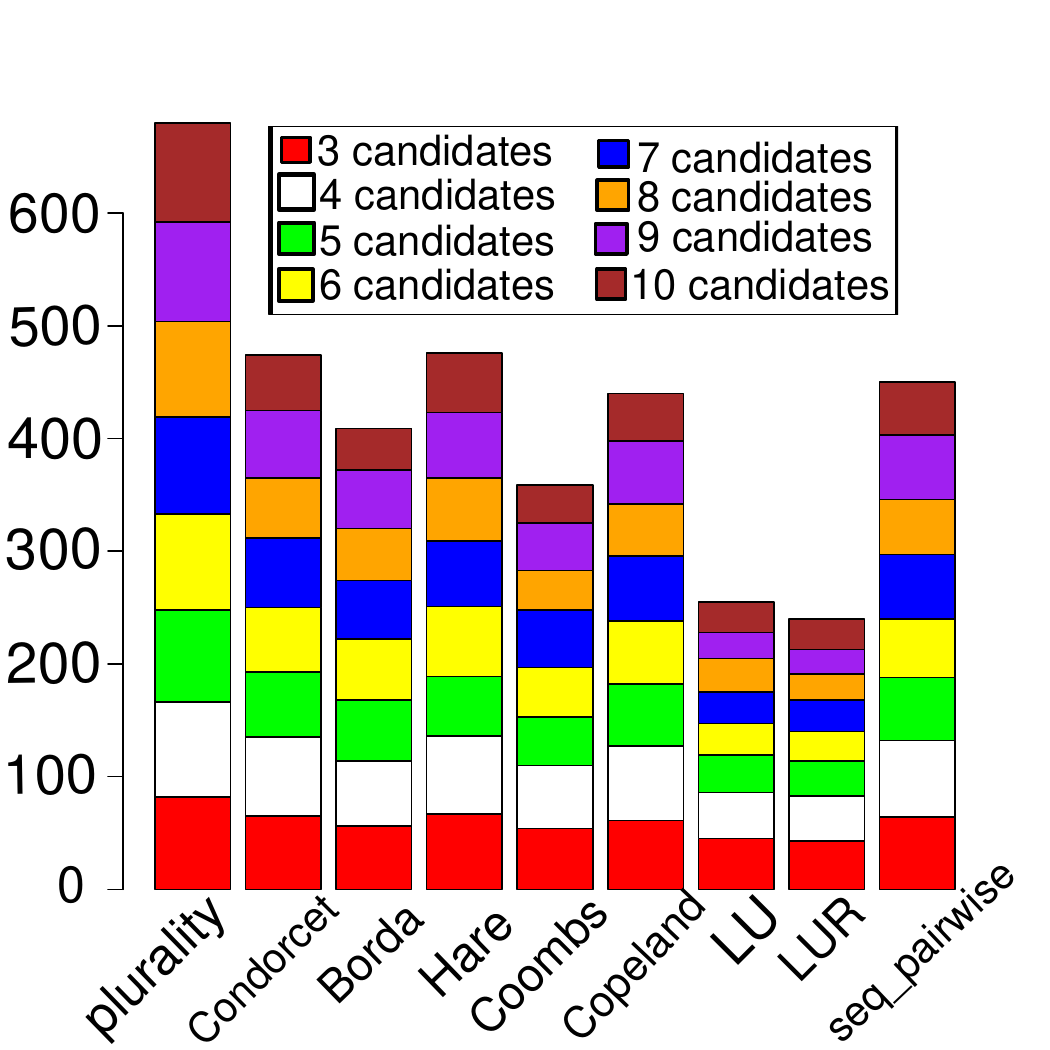}
	\caption{Performances of voting procedures against manipulation in the context of social media influencer. LU and LUR have the best performance. Except for plurality, there is no meaningful difference between other remaining rules.
	The number of affected elections is shown with different colors corresponding to the number of candidates.
	} \label{fig:influence}
    \end{figure}
\end{enumerate}
Compared to seven well-known voting rules, all of experimental results indicate that LU and especially LUR' voting protocols are considerably better and more resistant to manipulations. The implication of this study is that election experts who have been advocating the alternative vote/instant runoff may advocate also LU and especially LUR rules to avoid social disappointment and electing extremist candidate with substantial first round but little overall support in elections.

\section{Threats to validity}
In our experiments, we increase representativeness
by carefully designing scenarios to model manipulation in elections, using a published approaches explained in \cite{FaliszewskiRothe} and combining it with a custom approach for generating random profiles based on the insight from real elections. Nonetheless, our results must be interpreted within the constraints of how those models were generated. In fact, the use of synthetic data and
a controlled election environment does not consider several other scenarios, such as influence of diversity in district-based elections \cite{LewenbergLevRosenschein} or the impact of weighted voting on social disappointment.

For evaluations with real elections, voter diversity cannot be excluded and may affect the results, even though we carefully established
ground truth by measuring the performance of voting procedures on uniform randomness and repeating random elections many times.
We implemented the approaches according to the description provided
in the literature \cite{FaliszewskiRothe} and we set the parameters according to the recommendations provided by the authors, but cannot
exclude small differences due to implementation constraints. To take into consideration the randomness of profiles, we generated random profiles many times and calculated the absolute frequencies of all affected elections.
Despite confidence from synthetic data, the reader must be careful when generalizing results beyond the studied conditions. For example, the results cannot be generalized to weighted voting or non-uniform societies.

\section{Conclusions}
In this paper, we proposed two new concepts called social frustration and social disappointment-both to make less likely the election of candidates with limited overall support and to mitigate ethnic conflict and polarization in divided societies.  
We have also designed two new voting rules to prevent social disappointment in elections. In addition, a version of the impossibility theorem stated and proved regarding social disappointment in elections, showing that there is no voting rule for four or more candidates that simultaneously satisfies avoiding social disappointment and Condorcet winner criteria. Finally, we empirically evaluated the occurrence of social disappointment and we showed that the performance of our proposed protocols is superior to that of seven other well-known voting rules against manipulation in four different scenarios.

\begin{acknowledgements}
If you'd like to thank anyone, place your comments here
and remove the percent signs.
\end{acknowledgements}

%
%

\bibliographystyle{spbasic}      
\bibliography{sample}   


\end{document}